\documentclass[twocolumn,showpacs,preprintnumbers,amsmath,amssymb]{revtex4}
\usepackage{float}
\usepackage{graphicx}

\usepackage{braket}
\usepackage{graphicx}% Include figure files
\usepackage{dcolumn}% Align table columns on decimal point
\usepackage{bm}% bold math
\usepackage{amsthm}
\usepackage{amsmath}
\usepackage{amssymb}
\newtheorem{theorem}{Theorem}
\newtheorem{lemma}{Lemma}
\newtheorem{definition}{Definition}

\newtheorem{remark}{Remark}

\usepackage{algorithm}
\usepackage{algorithmic}
\usepackage{caption}
\usepackage[colorlinks,linkcolor=blue]{hyperref}

\theoremstyle{definition}
\theoremstyle{definition}
\newcommand{\DOI}[1]{doi: \href{https://doi.org/#1}{#1}}

\begin{document}
	\title{\bf{Characterization of exact one-query quantum algorithms }}
	\author{Weijiang Chen$^{1, 2, \#}$, Zekun Ye$^{1, 2, \#}$, Lvzhou Li$^{1, 3, }$ }\thanks{lilvzh@mail.sysu.edu.cn. \\  $^{\#}$Equal contribution, both the first author.}%
	
	\affiliation{%
		$^1$ Institute of Computer Science Theory, School of Data and Computer Science, Sun Yat-Sen University, Guangzhou 510006, China
	}%
	
	\affiliation{%
		$^2$ Center for Quantum Computing, Peng Cheng Laboratory, Shenzhen 518055, China
	}%
	\affiliation{%
		$^3$ Ministry of Education Key Laboratory of Machine Intelligence and Advanced Computing (Sun Yat-sen University), Guangzhou  {\rm 510006}, China
	}%
	
%	\date{}
	
	\begin{abstract}
		The quantum query models is one of the most important  models in quantum computing. Several well-known quantum algorithms are captured by this model, including the Deutsch-Jozsa algorithm, the Simon algorithm, the Grover algorithm and others. In this paper, we characterize the computational power of exact one-query   quantum algorithms. It is proved that  a total  Boolean function $f:\{0,1\}^n \rightarrow \{0,1\}$  can be exactly computed by a one-query quantum algorithm if and only if $f(x)=x_{i_1}$ or ${x_{i_1} \oplus x_{i_2} }$ (up to isomorphism). Note that unlike most work in the literature  based on the polynomial method, our proof does not resort to any knowledge about the polynomial degree of $f$.  \newline
		
		\noindent \DOI{10.1103/PhysRevA.101.022325} - Published 20 February 2020 
	\end{abstract}

	\maketitle

	\section{ Introduction}

	The classical decision tree models  have been well studied in classical computing,  and focus on problems such as the following: given a Boolean function $f:\{0,1\}^n\rightarrow \{0,1\}$,  how can we make  as few queries as possible  to the bits of $x$ in order to output the value of  $f(x)$? Quantum analogs, called quantum query models, have also attracted much attention in recent years \cite{Buhrman2002Complexity}. The implementation procedure of a quantum query model is  a  quantum query algorithm, which can be roughly described as follows:  it starts with a fixed state $|\psi_0\rangle$, and then performs the sequence of operations $U_0, O_x, U_1,  \ldots, O_x,U_t$, where $U_i$'s are unitary operators that do not depend on the input $x$ but the query $O_x$  does. This leads to the final state $ |\psi_x\rangle=U_tO_xU_{t-1}\cdots U_1O_xU_0|\psi_o\rangle$. The result is obtained by measuring the final state  $ |\psi_x\rangle$.

	The quantum query model can be discussed in two main settings: the exact setting and the bounded-error setting. A quantum query algorithm is said to compute a function $f$ exactly, if its output equals $f(x)$ with probability 1, for all inputs $x$.  It is said to {\it compute $f$ with bounded error},  if its output equals $f(x)$ with a  probability greater than a constant, for all inputs $x$. Roughly speaking, the query complexity of a function $f$  is the number of queries  that an  optimal (classical or quantum) algorithm should make in the worst case to compute $f$. The classical deterministic query complexity of  $f$ is denoted by $D(f)$, and the quantum query complexity in the  exact setting is denoted by $Q_E(f)$. In this paper, we focus on quantum query algorithms in the exact setting, which have been studied in much work \cite{Ambainis2013Superlinear, Ambainis2015Exact, Deutsch1992Rapid, Midrijanis2004Exact, Mihara2003Deterministic, He2018Exact, Montanaro2015On, Ambainis2013Exact, Ambainis2017Exact, Cai2018Optimal, Ambainis2016Superlinear, Cleve1998Quantum, Vasilieva2006Computing, Aaronson2016separations, Ambainis2017separations, Farhi1998alimit, Hayes2002quantum, Mischenko2015quantum, Braunstein2007exact, brassard1997exact, Qiu2018Generalized}. And quantum advantages were shown by comparing $Q_E(f)$ and $D(f)$. For total Boolean functions, Beals et al. \cite{beals2001quantum} showed that exact quantum query algorithms can only achieve  polynomial speed-up over classical counterparts. On the other hand, Ambainis et al. \cite{Ambainis2015Exact} proved that exact quantum algorithms have advantage for almost all Boolean functions. However, the biggest gap between $Q_E(f)$ and $D(f)$ is only a factor of 2 and is achieved by Deutsch's algorithm for a long time. In 2013,  a breakthrough result  was obtained by Ambainis, showing   the first total Boolean function for which exact quantum algorithms have superlinear advantage over  classical deterministic algorithms \cite{Ambainis2013Superlinear}. Moveover, Ambainis  \cite{Ambainis2016Superlinear} improved this result and presented a nearly quadratic separation in 2016.  For partial functions,  exponential separations between exact quantum and classical deterministic query complexity were obtained in several papers \cite{Deutsch1992Rapid,Shor1994Discrete,Simon1997On, Cai2018Optimal}. A typical example is the Deutsch-Jozsa algorithm \cite{Deutsch1992Rapid}.
	
	In this paper, we consider the following problem: what functions can be computed exactly by  one-query quantum algorithms (that can  make only one  query) ? Our motivation comes from the following two aspects:
	
	(i) Characterizing the computational power of a quantum computing model (or a kind of quantum algorithm) is of fundamental interest in the context of quantum 
complexity theory, and also is critical for  discovering quantum advantage.  Recently, characterization of one-query quantum algorithms in the bounded-error case has been considered by Aaronson et al. \cite{Aaronson2016Polynomials} and Arunachalam et al. \cite{Arunachalam2019Quantum}. But their results are not applicable to the exact case.
	
	(ii) Actually, the well-known Deutsch algorithm and Deutsch-Jozsa algorithm  belong to the class of exact one-query quantum algorithms. Then it is natural to ask what kind of functions (problems) can be computed  exactly by one-query quantum algorithms.

	We show that a total Boolean function $f$ can be computed exactly by a one-query quantum algorithm if and only if  $f(x)=x_{i_1}$ or $ x_{i_1} \oplus x_{i_2} $ (up to isomorphism). It is worth pointing out that unlike most work in the literature  based on the polynomial method, our proof does not depend on any knowledge about the polynomial degree of $f$. We hope this will illuminate a more general problem: what functions can be computed exactly by  $k$-query quantum algorithms?
	
	The remainder of this paper is organized as follows. The query models and the problem we consider are given in Section \ref{Pre}. The main results of this paper are presented in Section \ref{Result}. Finally, a conclusion is made in Section \ref{Con} and some further problems are proposed.
	
	\section{Preliminaries} \label{Pre}
	In this paper,  we consider  Boolean functions	$ f:\{0,1\}^n \rightarrow \{0,1\}$. Without special explanation, a function always means a total function, that is, it is defined for all $x\in\{0,1\}^n $. We will also refer to  partial functions that are defined on a subset $D\subset \{0,1\}^n$. Throughout this paper, a function is assumed  to be nonconstant, since the query complexity of a constant function is trivially zero. In the following, we first give an introduction about the query models, including both  classical and quantum cases, and then we describe  the problem  to be discussed.

	Given a Boolean function $ f:\{0,1\}^n \rightarrow \{0,1\}$, suppose $x=x_1x_2...x_n \in \{0,1\}^n $ is an input of $f$ and we use $x_i$ to denote its $i$-th bit. The goal of a query algorithm is to compute $f(x)$, given queries to the bits of $x$.
	
	\begin{figure}[ht]
		\centering
		\includegraphics[width=0.4\linewidth]{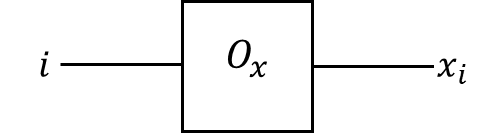}
		\caption{classical oracle}\label{CO}
		\label{classical oracle}
	\end{figure}
	
	In the classical case, the process of querying to $x$ is implemented by using the black box, which we call query oracle, as shown in Figure \ref{CO}. We want to compute $f(x)$ by using the query oracle as little as possible. 	A classical deterministic algorithm for computing $f$  can be described by a {\it decision tree}. 
	For example, suppose that we want to use a classical deterministic algorithm to compute  $f(x)=x_1 \land (x_2 \vee x_3)$.  Then  a decision tree $T$ for that  is depicted  in Figure \ref{decision tree}. 	Given an input $x$, the tree is evaluated as follows. It starts at the root. At each node, if it is a leaf, then  its label is output as the result for $f(x)$; otherwise, it queries its label variable  $x_i$. If $x_i$ = 0, then we recursively evaluate the left subtree. Otherwise, we recursively evaluate the right subtree. The query complexity of tree $T$ denoted by $D(T)$ is its depth, and we have $D(T)=3$ in this example. Given $f$, there exist different decision trees to compute it. The query complexity of $f$, denoted by $D(f)$, is defined as 
	\begin{equation}  \nonumber
	D(f)=\min_TD(T).
	\end{equation}

	\begin{figure}[H]
		\centering
		\includegraphics[width=0.3\linewidth]{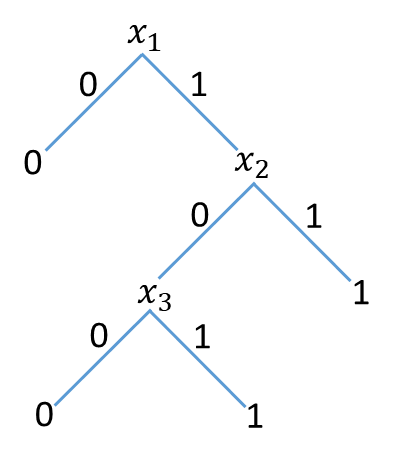}
		\caption{A decision tree $T$ for computing $f(x)=x_1 \land (x_2 \vee x_3)$}
		\label{decision tree}
	\end{figure}

	In the quantum case, the oracle is defined as $O_x|i,b\rangle = |i,b \oplus x_i\rangle$, where $i \in \{1,...,n\}$, as shown in  Figure 	\ref{quantum oracle}. Note that in this case, we are able to query more than one bit each time due to quantum superposition. A $T$-query quantum algorithm can be seen as a sequence of unitaries $U_TO_xU_{T-1}O_x...O_xU_0$, where $U_i$'s are fixed unitaries and $O_x$ depends on $x$ (See Figure 	\ref{T-query quantum algorithm}).
	
	\vbox{}
	
	\begin{figure}[H]
		\centering
		\includegraphics[width=0.4\linewidth]{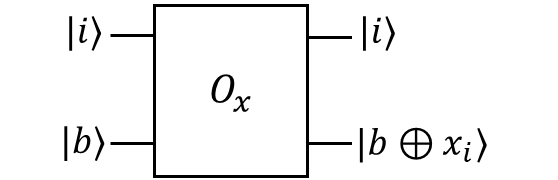}
		\caption{quantum oracle}
		\label{quantum oracle}
	\end{figure}

	\begin{figure}[H]
		\centering
		\includegraphics[width=0.8\linewidth]{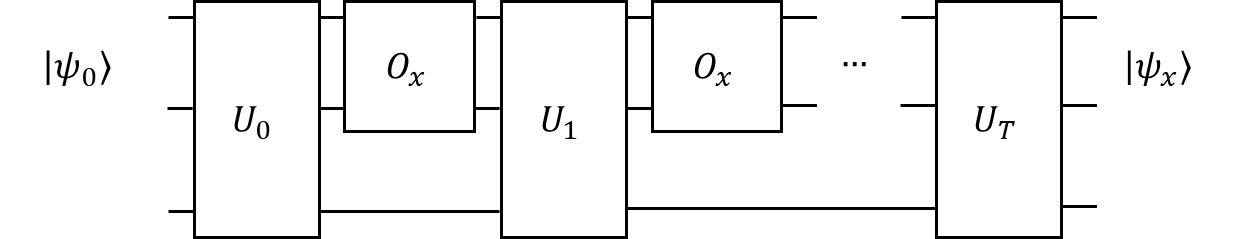}
		\caption{$T$-query quantum algorithm}
		\label{T-query quantum algorithm}
	\end{figure}

	The process of computation is as follows:
	\begin{itemize}		
		\item[ (1)]  Start with an initial state $|\psi_{0}\rangle$.
		
		\item[ (2)] Perform the operators $U_0, O_x, U_1, O_x...U_T$ in sequence, and then we obtain the state $|\psi_x\rangle=U_TO_xU_{T-1}O_x...U_0|\psi_{0}\rangle$.
		
		\item[ (3)] Measure $|\psi_x\rangle$ with a $0-1$ positive operator-valued measurement  \cite{Nielsen2002Quantum}. The measurement result
		is regarded as the output of the algorithm.
	\end{itemize}
	
	In the above,  we use $r(x)$ to denote the measurement result of $|\psi_x\rangle$. Let $P[\mathcal{A}]$ denote the probability that event $\mathcal{A}$ occurs.  If it satisfies:
	
	\begin{equation}  \nonumber
	\forall x , P[ r(x)=f(x) ]  \geq 1-\epsilon   , 
	\end{equation}
	where $\epsilon < \frac{1}{2} $, then the quantum query algorithm is said to compute $f(x)$ with bounded error $\epsilon$. If it satisfies:
	\begin{equation} \nonumber
	\forall x, P[ r(x)=f(x) ]  =1,
	\end{equation} 
	then it is said to compute $f(x)$ exactly.
	
	The exact quantum  query complexity of   $f$, denoted by $Q_E(f)$, is the minimum number of queries that a quantum query algorithm needs to compute $f$.
	The gap between $D(f)$ and $Q_E(f)$ is usually used to exhibit quantum advantage.

	In this paper, we want to characterize those functions $f$ that satisfy $Q_E(f)=1$. In other words, we consider this problem: what functions $f$ can be computed exactly by a one-query quantum algorithm?  In this case a quantum query algorithm is as shown  in Figure \ref{one-query}.
	\begin{figure}[H]
		\centering
		\includegraphics[width=0.8\linewidth]{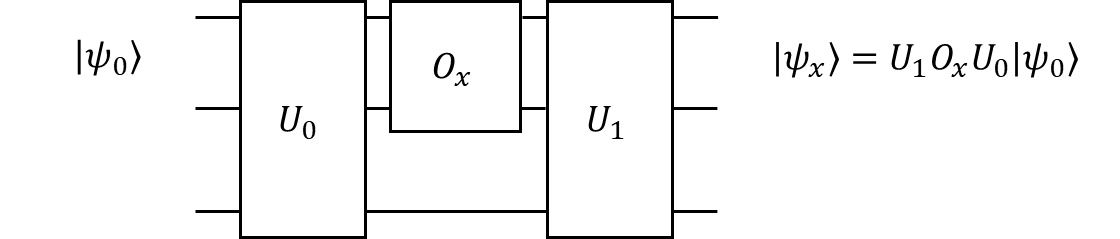}
		\caption{one-query quantum algorithm}
		\label{one-query}
	\end{figure}

	\section{Main result} \label{Result}
	%\subsection{total function}
	Two  functions $f$ and $g$ over $\{0, 1\}^n$  are {\em isomorphic}  if they are equal up to negations and permutations of
	the input variables, and negation of the output variable. It is easy to see that for any two isomorphic Boolean functions $f$ and $g$, we have $Q_E(f)=Q_E(g)$. Now our main result is as follows.
	
	\begin{theorem}
		\label{theorem1}
		A total Boolean function $f:\{0,1\}^n \rightarrow \{0,1\}$ can be  computed exactly by a one-query quantum algorithm, if and only if $f(x)=x_{i_1}$ or $x_{i_1} \oplus x_{i_2} $ (up to isomorphism).
		
	\end{theorem}
	
\begin{remark}
	In the above theorem, we  consider only total functions. A more interesting  problem is to characterize the partial functions (promise problems) that can be computed exactly by  one-query quantum algorithms. This problem seems to be more complicated, and the method  used here may not be applicable to  partial functions. 
	
	\end{remark}
	First, it is easy to see the sufficiency. If $f(x)=x_{i_1} \oplus x_{i_2} $, then it can be computed exactly by the Deutsch algorithm.  If $f(x)=x_{i_1}$, there obviously exists a one-query quantum algorithm to do that. Therefore, the key to prove Theorem \ref{theorem1} is the necessity. For that we will prove the following: (a) if $f$ can be exactly computed by a one-query quantum algorithm, then it depends on at most two variables, and (b) furthermore it must be in the above form.
	
	\begin{definition}
		\label{definition}
		For a total Boolean function $f:\{0,1\}^n \rightarrow \{0,1\}$, $f$ is said to depend on the $i$th variable of the input, if there exists $x \in \{0,1\}^n$ such that $f(x) \neq f(x^i)$, where $x^i$ is the same as $x$ except for the $i$-th bit being flipped.	
	\end{definition}
	For example, if $f(x)=x_{i_1} \oplus x_{i_2} $, then $f$ depends on the $i_1$-th and $i_2$-th variables.	
	
	%It means if we want to find an algorithm $A$ to compute $f$ exactly, we need to design $U_0$ to ensure that the system of equations (totally $C_0 \times C_1+1$ equations) has solutions. 
	%\begin{lemma}
	%For $p,q \in C$, if $|p - q|^2 = m$, then $|p|^2 + |q|^2 \geq \frac{m}{2}$.
	%\end{lemma}

	\subsection{Proof of necessity}
	%Firstly, we prove the sufficiency as follows.
	%Next, we prove the necessity. 

	Now suppose that $f:\{0,1\}^n \rightarrow \{0,1\}$ can be computed exactly by a one-query quantum algorithm as shown in Figure \ref{one-query}. We denote $$U_0|\psi_0\rangle = \sum_{i,j,k} \alpha_{ijk}|i\rangle|j\rangle|k\rangle, $$
	where 	$i\in\{1,2,\cdots,n\}$, $j\in\{0,1\}$, and $|k\rangle$ is an arbitrary ancilla register, corresponding to the third line in Figure 5. We have the following observation.
	\begin{lemma}
		For a pair of inputs $x,y$ such that $f(x) \neq f(y)$, let $S = \{i|x_i \neq y_i\}$. Then it holds that 
		\begin{equation}\sum_{i \in S}\sum_k|\alpha_{i0k} - \alpha_{i1k}|^2 = 1.\label{lm1}	\end{equation}
	\end{lemma}
	
	\begin{proof}[\textbf{Proof}]  
		Note that $|\psi_x\rangle = U_1 O_x U_0|\psi_0\rangle $ and  let $|\phi_x\rangle = O_x U_0|\psi_0\rangle$. Then we get
		$$	|\phi_x\rangle =O_x \sum_{i,j,k} \alpha_{ijk}|i\rangle|j\rangle|k\rangle = \sum_{i,j,k} \alpha_{ijk}|i\rangle|j \oplus x_i\rangle |k\rangle. $$  The assumption that $f$ can be computed exactly 	implies that $|\psi_x\rangle$ and $|\psi_y\rangle$ can be perfectly distinguished for $x,y$ satisfying $f(x)\neq f(y)$. Thus, the two states are mutually orthogonal, that is 
		$$ \langle \psi_x|\psi_y\rangle=0.$$ Since unitary operators do not change the orthogonality between two states, equivalently, there is
		\begin{equation*}
		\langle\phi_x|\phi_y\rangle = 0,\label{eq:cond}
		\end{equation*}
		from which it follows that
		\begin{align}
		\langle \phi_x|\phi_y \rangle& =\sum_{i,j,k} \alpha^{*}_{ijk}\langle i| \langle j \oplus x_i| \langle k|    \sum_{p,q,r} \alpha_{pqr} |p\rangle|q \oplus y_p\rangle|r\rangle \nonumber\\
		&= \sum_{i,j,k,q} \alpha^{*}_{ijk}\alpha_{iqk} \langle j \oplus x_i| q \oplus y_i \rangle = 0.\label{eq:2}
		\end{align}
		Furthermore, Eq. (\ref{eq:2}) can be rewritten as
		\begin{equation}
		\sum_{i \notin S}\sum_{k}(|\alpha_{i0k}|^2 + |\alpha_{i1k}|^2)+\sum_{i \in S}\sum_{k}(\alpha^{*}_{i0k} \alpha_{i1k}+\alpha^{*}_{i1k}  \alpha_{i0k}) = 0.\label{eq:3}
		\end{equation}
		Moreover, note that $\sum_{i,j,k}|\alpha_{ijk}|^2 = 1$, that is
		\begin{equation}
		\sum_{i \notin S}\sum_{k}(|\alpha_{i0k}|^2 + |\alpha_{i1k}|^2) + \sum_{i \in S}\sum_{k}(|\alpha_{i0k}|^2 + |\alpha_{i1k}|^2) = 1.\label{eq:4}
		\end{equation}
		Therefore, by subtracting Eq.	(\ref{eq:3}) from Eq. (\ref{eq:4}), we obtain
		$$\sum_{i \in S}\sum_{k}\left[(|\alpha_{i0k}|^2 + |\alpha_{i1k}|^2)-(\alpha^{*}_{i0k} \alpha_{i1k}+\alpha^{*}_{i1k} \alpha_{i0k})\right] = 1,
		$$
		which is equivalent to
		$$
		\sum_{i \in S}\sum_k|\alpha_{i0k} - \alpha_{i1k}|^2 = 1.
		$$ 
	\end{proof}
	Now we are ready for proving the necessity.
	\begin{lemma}
		\label{lemma2}
		If a total Boolean function $f:\{0,1\}^n \rightarrow \{0,1\}$ can be computed exactly by a one-query quantum algorithm, then $f$ depends on at most two variables of $x$, and furthermore $f$ is in the form $f(x)=x_{i_1}$ or $x_{i_1} \oplus x_{i_2} $ (up to isomorphism).
	\end{lemma}
	\begin{proof}[\textbf{Proof}] 
		We denote the Hamming distance of $x$ and $y$ by $d(x,y)$. For a total function $f$, if $f$ depends on some variable $x_i$, then there exist $x,y\in\{0,1\}^n$ such that $d(x,y) = 1, x_i \neq y_i$ and $f(x) \neq f(y)$. Thus, in this case $S=\{i\}$, and by Lemma \ref{lm1} we have 
		$$
		\sum_k|\alpha_{i0k}-\alpha_{i1k}|^2=1.
		$$
		Now suppose $f$ depends on $t$ variables $x_{i_1},x_{i_2},...,x_{i_t}$. It follows that
		$$
		\sum_{r=1}^{t}	\sum_k|\alpha_{i_r0k}-\alpha_{i_r1k}|^2=t.
		$$
		
		Furthermore, we have 
		\begin{equation*}\sum_{r=1}^{t}	\sum_k|\alpha_{i_r0k}-\alpha_{i_r1k}|^2 \leq 2\sum_{r=1}^{t}	\sum_k|\alpha_{i_r0k}|^2+|\alpha_{i_r1k}|^2\leq 2,\end{equation*}
		where the first equality follows from the observation that $|a-b|^2 \leq 2(|a|^2+|b|^2)$ for any two complex numbers $a, b$, and the second equality holds because we have  $\sum_{i,j,k}|\alpha_{ijk}|^2=1$.
		Thus,  we get $t\leq 2$, that is, $f$  depends on at most two variables.
		
		Next, we prove $f(x)=x_{i_1}$ or $ x_{i_1} \oplus x_{i_2}  $ (up to isomorphism). 
		Suppose that  $f$ depends on at most two variables $x_{i_1}$ and $x_{i_2}$. 
		Denote $C_{00} = \{ x|x_{i_1}=x_{i_2}=0\}, C_{01}=\{x|x_{i_1}=0,x_{i_2}=1\}, C_{10}=\{x|x_{i_1}=1,x_{i_2}=0\}, C_{11}=\{x|x_{i_1}=1,x_{i_2}=1\}$. Let $S_0=\{x|f(x)=0\}$ and $S_1=\{x|f(x)=1\}$. Without loss of generality, suppose $C_{00} \subseteq S_0, C_{10} \subseteq S_1$.  Below we show $f = x_{i_1} $ or $f = x_{i_1} \oplus x_{i_2}$.  Other cases can be discussed in a similar way.
		\begin{enumerate}
			\item[(1)] $C_{01} \subseteq S_0, C_{11} \subseteq S_1$. In this case, we have $f = x_{i_1} $.
			
			\item[(2)] $C_{01} \subseteq S_1, C_{11} \subseteq S_0$. In this case, we have $f = x_{i_1} \oplus x_{i_2}$.
			
			\item[(3)] $C_{01} \subseteq S_0, C_{11} \subseteq S_0$. In this case, $f=x_{i_1} \land \neg{x_{i_2}}$, which is isomorphic to $AND_2$  that can not be computed by one-query quantum algorithms as indicated in \cite{Ambainis2015Exact}. Thus, this case is impossible.
			
			\item[(4)] $C_{01} \subseteq S_1, C_{11} \subseteq S_1$. In this case, $f=x_{i_1} \vee x_{i_2} $, which is also isomorphic to $AND_2$. Thus, this case is  impossible.
		\end{enumerate}

	\end{proof}

	\section{Conclusion} \label{Con}
	In this paper we have characterized the power of exact one-query quantum algorithms for total functions. In conclusion, a total Boolean function $f:\{0,1\}^n  \rightarrow \{0,1\}$ can be  computed by exactly a one-query algorithm if and only if
	$f(x)=x_i$ or  $x_{i_1} \oplus x_{i_2} $  (up to isomorphism). Note that unlike most work in  the  literature  based  on  the  polynomial  method,  our proof doe
	s not resort  to any knowledge about the polynomial  degree  of $f$.  We hope it will illuminate two more general problems that are worthy of further consideration as follows.
	
	\textbf{A}. \textbf{Characterization of partial  functions that can be computed exactly by  a one-query quantum algorithm.} 
	Given a partial Boolean function $f:D \rightarrow \{0,1\}$, where $D \subset \{0,1\}^n $, how do we determine whether there exists a one-query quantum algorithm that computes it exactly? Furthermore, can we discover all those partial functions $f$ that can be computed exactly by a one-query quantum algorithm?
	
	\textbf{B}. \textbf{Characterization of functions that can be computed exactly by a $k$-query quantum algorithm.}
	A more interesting problem is to discuss the power of exact $k$-query quantum algorithms, although for some functions with a specific property, we know their quantum query complexity. There is no a general conclusion to characterize the power of exact $k$-query quantum algorithms. Figuring out this problem is useful for further understanding quantum query algorithms and inspiring us to find more problems with quantum advantage.
	
	\section*{Acknowledgements}{This work is supported by the National Natural Science Foundation of China (Grant No. 61772565), the Natural Science Foundation of Guangdong Province of China (Grant No. 2017A030313378), the Science and Technology Program of Guangzhou City of China (Grant No. 201707010194), the Key Research and Development  project of Guangdong Province (Grant No. 2018B030325001) and the Open Research Fund from State Key Laboratory of High Performance Computing of China (Grant No. 201901-03).
	}
	
	\bibliographystyle{unsrt}

\end{document}